\documentclass{llncs}
\usepackage{makeidx}
\usepackage{graphicx}             
\usepackage{amsmath}       
\usepackage[latin1]{inputenc}  
\usepackage{subfigure}
\usepackage{algorithm}
\usepackage[noend]{algorithmic}
\usepackage{amssymb}  
\usepackage[numbers,sectionbib]{natbib}
\newcommand{\lc}{\left\{}  
\newcommand{\rc}{\right\}}

\newcommand{\lb}{\left(}  
\newcommand{\rb}{\right)} 
 
\newcommand{\Rcal}{\mathcal{R}} 
\newcommand{\Mcal}{\mathcal{M}} 
\newcommand{\Ical}{\mathcal{I}} 
\newcommand{\Jcal}{\mathcal{J}} 
\newcommand{\bb}{\backslash} 

\newcommand{\ba}{\[\begin{aligned}}
\newcommand{\ea}{\end{aligned}\]}
\newcommand{\bi}{\begin{itemize}}
\newcommand{\ei}{\end{itemize}}
\newcommand{\SM}{\textsc{SM}}
\newcommand{\SMG}{\textsc{SMG}}
\newcommand{\SMTI}{\textsc{SMTI}}
\newcommand{\cSM}{\textsc{c3DSM}}
\newcommand{\SE}{\textsc{SE}}

\begin{document}
\title{Stable marriage with general preferences}

\author{Linda Farczadi,  Konstantinos Georgiou,
Jochen K\"{o}nemann}

\institute{Dept. of Combinatorics and Optimization,\\ University of Waterloo, Waterloo, Canada \\
\email{\{lfarczadi,k2georgi,jochen\}@uwaterloo.ca}}

\maketitle

\begin{abstract} We propose a generalization of the classical stable marriage problem. In our model, the preferences on one side of the partition are given in terms of arbitrary binary relations, which need not be transitive nor acyclic. This generalization is practically well-motivated, and as we show, encompasses the well studied hard variant of stable marriage where preferences are allowed to have ties and to be incomplete.  As a result, we prove that deciding the existence of a stable matching in our model is NP-complete. Complementing this negative result we present a polynomial-time algorithm for the above decision problem in a significant class of instances where the preferences are asymmetric. We also present a linear programming formulation whose feasibility fully characterizes the existence of stable matchings in this special case. 

Finally, we use our model to study a long standing open problem regarding the existence of cyclic 3D stable matchings. In
particular, we prove that the problem of deciding whether a fixed 2D perfect matching can be extended to a 3D stable matching is \textsc{NP}-complete, showing this way that a natural attempt to resolve the existence (or not) of 3D stable matchings is bound to fail. 
\end{abstract}

\section{Introduction}

The Stable Marriage (\textsc{SM}) problem is a classical bipartite matching problem first introduced by Gale and Shapley
\cite{gale1962college}.  An instance of the problem consists of a set of $n$ men, and a set of $n$ women. Each man (woman) has a preference list that is a total order over the entire set of women (men). The goal is to find a stable matching between the men and women, meaning that there is no (man, woman) pair that both prefer each other to their current partners in the matching.  Since its introduction, the stable marriage problem has become one of the most popular combinatorial problems with several books being dedicated to its study \cite{gusfield1989stable, knuth1976mariages, roth1992two} and more recently \cite{manlove2013algorithmics}.  The popularity of this model arises not only from its nice theoretical properties but
also from its many applications. In particular, a wide array of allocation problems from many diverse fields can be analyzed within its context. Some well known examples include the labour market for medical interns, auction markets, the college admissions market, the organ donors market, and many more \cite{roth1992two}.

In their seminal work Gale and Shapley showed that every instance of \textsc{SM} admits a solution and such a solution can be 
computed efficiently using the so-called Gale-Shapley (or man-proposing) algorithm. Among the many new variants of this classical problem, two extensions have received most of the attention: incomplete preference lists and ties in the 
preferences. Introducing either one of these extensions on its own does not pose any new challenges, meaning that solutions are still guaranteed to exist, all solutions have the same size, and they can be computed using a modification of the original Gale-Shapley algorithm \cite{gale1985some, gusfield1989stable}. However, the same cannot be said about the Stable Marriage problem with Ties and Incomplete Lists (\textsc{SMTI}) that incorporates both extensions. In this variant stable matchings 
no longer need to be of the same size, even though they are still guaranteed to exist. In fact, deciding whether a given instance admits a stable matching of a given size is $\textsc{NP}$-hard \cite{hard}, even in the case where ties occur only on one side of the partition. Several papers have studied the approximate variants of this problem (see \cite{manlove2013algorithmics} for a more complete account).

A central assumption in most variants of \textsc{SM} is that agents' preferences are transitive (i.e., if $x$ is preferred to $y$, and $y$ is preferred to $z$ then $x$ is also preferred to $z$). However, there are several studies \cite{birnbaum2008experimental, fishburn1991nontransitive, brams2009mathematics, may1954intransitivity} that suggest that non-transitive, and even cyclic preferences arise naturally. Cyclicity, for example, may be introduced in the context of multi-attribute comparisons \cite{ fishburn1999preference}; e.g., consider the following study from \cite{may1954intransitivity} where $62$ college students were asked to make binary comparisons between three potential marriage partners $x$,$y$ and $z$ according to the following three criteria: intelligence, looks and wealth. The candidates had the following attributes: candidate $x$ was very intelligent, plain, and well off; candidate $y$ was intelligent, very good looking, and poor; and
candidate $z$ was fairly intelligent, good looking, and rich. From the $62$ participants, $17$ displayed the following cyclic preference: $x$ was preferred to $y$, $y$ was preferred to $z$, and $z$ was preferred to $x$. In order to better capture such situations there is a need for a model that allows for more general preferences.

Addressing this need we propose the Stable Marriage with General Preferences (\SMG) problem. As in \SM, in an instance of \SMG\ we are given $n$ men, and $n$ women, and the preferences of men are complete total orders over the set of women. The preferences of women, on the other hand, are given in terms of arbitrary binary relations over the men. Each of these binary relations will be represented by a set of ordered pairs of men. We say that a woman prefers man $x$ at least as much as man $y$ if the ordered pair $(x,y)$ is part of her preference set. A matching is then stable as long as for every unmatched (man, woman) pair at least one member prefers her mate in the matching at least as much as the other member of the pair.

This introduction of non-transitive preferences, even when restricted to just one side of the partition, changes the properties
of the model drastically. Like in \SM\ any solution must be a perfect matching. However, solutions are no longer guaranteed to exist. We show that non-transitive preferences generalize both incomplete lists and ties by reducing the \SMTI\ to \SMG. In doing so, we prove that the \SMG\ problem is also \textsc{NP}-hard. In addition, we provide results on the structural properties of the \SMG\ problem and give sufficient conditions for the problem to be solvable in polynomial time.

The second half of this paper focuses on three-dimensional stable matching models whose study was initiated by Knuth~\cite{knuth1976mariages}. We will be particularly interested in the Cyclic 3-Dimensional Stable Matching problem (\cSM), where we are given a set of $n$ men, a set of $n$ women and in addition a set of $n$ dogs. The preferences of the men are complete total orders over the set of women. Similarly the women have preferences over the dogs, and the dogs have preferences over the men. A 3D matching is said to be stable if there is no (man, woman, dog) triple that is strictly preferred to their current triples in the matching by each of its members. A prominent open question is whether every instance of \cSM\ admits a stable matching, and whether it can be computed efficiently. 

A natural avenue for attacking \cSM\ is to solve the following problem which we refer to as Stable Extension (\SE): suppose we fix a perfect matching $M$ on dogs and men, can we efficiently determine whether $M$ is extendible to a 3D stable matching? Recall that women have preferences over dogs only, but note that the given matching $M$ induces preferences over their male owners as well! In essence, this allows us to state the \SE\ problem as a two dimensional bipartite matching problem, and we show in Theorem \ref{thm0} that \SE\ can be seen as a special case of \SMG.  We then prove that \SE\ remains \textsc{NP}-complete.

\medskip
\noindent \textbf{Contributions.} In Section \ref{sec3} we show the following result.

\begin{theorem}\label{thm1} \SMG\ is \textsc{NP}-complete. 
\end{theorem}
We then identify a significant class of instances that are solvable in polynomial time: those where the preferences are asymmetric, meaning that for every pair of men $x,y$, each woman prefers at most one to the other. We then prove the following result.
\begin{theorem}\label{thm2} For instances of \SMG\ with asymmetric preferences, there exists a polynomial time algorithm 
that finds a solution if and only if one exists.
\end{theorem}
We provide two different proofs. The first (given in Section \ref{sec3.2}) employs an adaptation of the classical Gale-Shapley man-proposing algorithm. The second (given in Section \ref{sec3.3}) relies on a polyhedral characterization: we define a polytope that is non-empty if and only if the instance admits a stable matching. We also develop an efficient rounding algorithm for its fractional points. Despite displaying stronger structural properties than \SMG, we show that \SE\ remains hard to solve.
\begin{theorem}\label{thm3} \SE\ is \textsc{NP}-complete.
\end{theorem}
The proof of the above theorem is given in Section \ref{sec4}. At a high level, its strategy resembles that of the proof of
Theorem \ref{thm1}. The details are however significantly more intricate, mainly due to the fact that \SE\ instances correspond to \SMG\ instances in which preferences are induced by a given 3D matching instance. As an interesting consequence for the \cSM, Theorem \ref{thm3} rules out the natural algorithmic strategy of fixing and extending a 2D perfect matching on two of the input sets.  \\

\noindent \textbf{Related work.} To the best of our knowledge, the stable marriage problem with preferences  given in terms of arbitrary binary relations has not been studied before. In \cite{abraham2003algorithmics} the authors do consider a version of \SM\ with non-transitive preferences, however unlike in our model, the preference relations are required to be acyclic. The authors do not study this problem directly but instead use it as a tool for developing a reduction between the stable roommates problem and the stable marriage problem.  

There is a rich literature on \SM\ and its variants. In particular, there has been significant work on the following approximate variant of \SMTI: given an instance of SMTI, find a  {\em maximum size} stable matching. When ties are allowed on both sides, the problem is  \textsc{NP}-hard to approximate within $33/29$  \cite{yanagisawa2007approximation} and the currently best known ration is $3/2$ \cite{mcdermid20093}. When ties are only allowed on the side of the women the problem is \textsc{NP}-hard to approximate within  $21/19$ \cite{halldorsson2007improved} and the currently best known ratio is $25/17$ \cite{iwama201425}.

A related model known as Stable Marriage with Indifference \cite{irving1994stable, manlove2002structure}, allows for preferences to be given in the form of partial orders, that are not necessarily expressible as a single list involving ties. That is, the indifference relation need not be transitive.    This model allows for several definitions of stability, and depending on which definition is used, solutions might not always exist. In particular, under strong stability, stable matchings might not always exists and it is show in \cite{irving2003strong} that the problem of deciding whether a strongly stable matching exists, given an instance of the stable marriage problem with partially ordered preferences, is NP-complete.

For \cSM, it is known that every instance admits a stable matching for $n \leq 4$ \cite{eriksson2006three}. The authors 
conjectured that this result can be extended to general instances.  In \cite{biro2010three} it was shown that if we allow unacceptable partners, the existence of a stable matching becomes \textsc{NP}-complete. In the same paper, and also independently in \cite{huang2010circular}, it was shown that the \cSM\ problem under a different notion of stability 
known as strong stability is also \textsc{NP}-complete. 

\section{Definitions and notation}

Throughout this paper we denote the set of men by $B$ and the set of women by $C$. In the 3D setting, we have an additional set of $n$ dogs $A$. A 3D perfect matching $\Mcal$ is a set of $n$ disjoint triples from $A \times B \times C$.  For  every dog $a \in A$ we denote by $\Mcal(a)$ the man that $a$ is matched to in $\Mcal$. Similarly for every man $b \in B$, $\Mcal(b)$ denotes  the woman that $b$ is matched to in $\Mcal$, and for every woman $c \in C$, $\Mcal(c)$ denotes the dog that $c$ is matched to in $\Mcal$. A 3D perfect matching can also be induced by fixing perfect matchings on any two of the following sets  $A \times B$, $B \times C$ or $C \times A$. In particular, we will use $M$ to denote a perfect matching on $A \times B$ and $N$ to denote a perfect matching on $B \times C$. We then define the 3D matching $M \circ N$ by setting $(a,b,c) \in M \circ N$  if and only if $(a,b) \in M$ and $(b,c) \in N$. For each $q \in A \cup B$ we denote by $M(q)$  the partner of $q$ in $M$, and similarly for each $q \in  B \cup C$ we denote by $N(q)$ the partner of $q$ in $N$.  If the preferences of an agent $q$ are given in terms of an ordering $P(q)$ (with or without ties) over a set $A$ then for all $x,y \in A$ we write $x \succ_q y$ to denote that $q$ strictly prefers $x$ to $y$ and $x \succeq_q y$ to denote that $q$ prefers $x$ at least as much as $y$.
 
\subsection{Stable marriage with ties and incomplete lists (\SMTI) }\label{sec2.1}

An instance $\Ical$ of \SMTI\ consists of a set $B$ of $n$ men and a set $C$ of $n$ women. Each man $b \in B$  has a preference list $P(b)$ that is an ordering over a subset of $C$ and is allowed to contain ties. Similarly each woman $c \in C$  has a preference list $P(c)$ that is an ordering over a subset of $B$ and is also allowed to contain ties.  A pair $(b,c)$ is said to be \textit{acceptable} if $b$ appears in $P(c)$ and $c$ appears in $P(b)$.  It is assumed that a woman $c$ is acceptable to a 
man $b$ if and only if man $b$ is acceptable to woman $c$. A pair $(b,c)$ is \textit{blocking} with respect to a matching $N$ if $(b,c)$ is an acceptable pair that is not in $N$ , $c \succ_b N(b)$ and $b \succ_c N(c)$. A matching $N$ is \textit{stable} if it uses only acceptable pairs and it has no blocking pairs. In that case, we also say that $N$ is a solution to $\Ical$. If a solution is of size $n$ we refer to it as a perfect stable matching. While all instances of \SMTI\ admit a stable matching, not all instances admit a perfect stable matching. In this paper we use \SMTI\ to refer to the decision problem of whether a given instance admits a perfect stable matching. This problem is known to be \textsc{NP}-complete \cite{hard}, even when the ties occur only in the preference lists of the women.

\subsection{Stable marriage with general preferences (\SMG)}\label{sec2.2}

An instance $\Ical$ of \SMG\ consists of a set $B$ of $n$ men and a set $C$ of $n$ women. Each man $b \in B$ has a preference list $P(b)$ that is complete total order over $C$.  Each woman $c \in C$ has a preference relation given in terms of a set of ordered pairs $\Rcal_c \subseteq B \times B$. For a given pair of men $b,b' \in B$ and woman $c \in C$ we interpret $(b,b') \in \Rcal_c$ as  woman $c$ preferring man $b$ at least as much as man $b'$. Note that whether $(b,b')$ is in $\Rcal_c$ 
is completely independent of whether $(b',b) \in \Rcal_c$. We say that a pair $(b,c)$ is \textit{blocking} with respect to a matching $N$, if $b$ and $c$ are not matched to each other and neither one  prefers its partner in $N$ at least as much as the other. Formally, $(b,c)$ is blocking if  $(b,c) \notin N$, $c \succ_b N(b)$ and $( N(c),b) \notin \Rcal_c$. A matching $N$ is \textit{stable} if it has no blocking pairs. It follows from this definition that any stable matching is a perfect matching. In this paper we use \SMG\ to refer to the decision problem of whether a given instance admits a  stable matching (which we also call a solution). An instance $\Ical$ of \SMG\ is said to have \textit{asymmetric preferences} if for every $b_1,b_2 \in B$ and $c \in C$ at most one of the following two conditions holds: $(b_1,b_2) \in \Rcal_c$ or $(b_2,b_1) \in \Rcal_c$.

Note that we could have obtained an alternate definition of stability by saying that a pair $(b,c)$ is blocking if $(b,c) \notin N$, $c \succ_b N(b)$ and $(b,N(c)) \in \Rcal_c$. However, the two models are equivalent via the following correspondence: create a new instance $\Ical'$ with sets $\Rcal'_c$ where $(b,b') \in \Rcal'_c$ if and only if $(b',b) \notin \Rcal_c$. Then the solutions that are stable for $\Ical$ under the definition of stability used in this paper, are exactly those that are stable for $\Ical'$ using the alternate definition of stability. Hence, we can use our definition of stability without loss of generality.

It is easy to see that any stable matching for an \SMG\ instance must be a perfect matching. This contrasts the case of \SMTI\ where a given instance can have stable matchings of different size. A second difference is that unlike in \SMTI, not all instances of \SMG\ admit a stable matching. This can be observed from the following example: suppose there are two men $b_1,b_2$ and two women $c_1,c_2$. Both men prefer woman $c_1$ to woman $c_2$ and woman $c_1$ has a preference relation given by the set $R_{c_1} = \emptyset$.  Then given any perfect matching, woman $c_1$ will always form a  blocking pair with the man that she is not matched to. Hence this instance does not admit a stable matching.

\subsection{Stable extension (\textsc{SE})}\label{se}

An instance $\Ical$ of \SE\ consists of a set of $n$ dogs $A$, a set of $n$ men $B$, and a set of $n$ women $C$,  together with a fixed perfect matching $M$ on $A \times B$.  The preferences are  defined cyclically ($A$ over $B$, $B$ over $C$, and $C$ over $A$) and are complete total orders over the corresponding sets.  A triple $(a,b,c)$ is \textit{blocking} with respect to a 3D matching $\Mcal$ if $(a,b,c) \notin \Mcal$, $b \succ_a \Mcal(a)$, $c \succ_b \Mcal(b)$ and $a \succ_c \Mcal(c)$. Note that if $(a,b,c)$ is a blocking triple then $a,b$ and $c$ must be part of three disjoint triples in $\Mcal$.   A 3D matching $\Mcal$ is stable if it has no blocking pairs.  It follows from this definition that any stable 3D matching must be a perfect matching.  We say that a perfect matching $N$ on $B \times C$ is a \textit{stable extension}, or a solution to $\Ical$, if $M \circ N$ is a 3D stable matching, and we use \SE\ to refer to the decision problem of whether a given instance admits a stable extension.

We now demonstrate how an instance $\Ical$ of \SE\ can be reduced to an \SMG\ instance. First, for each man $b \in B$ we define $A_b$ to be the set of dogs in $A$ that prefer $b$ to the man assigned to them in the fixed perfect matching $M$. That is $A_b = \lc a \in A : b \succ_a M(a) \rc$. The set $A_b$ contains exactly those dogs in $A$ with whom man $b$ can potentially be in a blocking triple when extending $M$ to a 3D matching. It follows that if $A_b = \emptyset$ then man $b$ cannot be in a blocking triple in any extension of $M$ to a 3D matching. Now, for each pair $(b,c)$ we define  $ \alpha(b,c)$ to be the dog in the set $A_b$ that woman $c$ prefers the most. That is $\alpha(b,c) = \max_{\succ_c} A_b$. If $A_b = \emptyset$ then we let $\alpha(b,c)$ be the dog in the last position in woman $c$'s preference list. We now define preferences for each woman $c \in C$ 
\begin{align}\label{set}
\Rcal_c := \lc (b,b') \,| \, b,b' \in B, b \neq b', \, M(b) \succeq_c \alpha(b',c) \rc.
\end{align}
Note that if $(N(c),b) \in \Rcal_c$ then in the 3D matching $M \circ N$ the woman $c$ will be matched to a dog that she 
prefers at least as much as any dog in $A_b$, therefore guaranteeing that the man $b$ and woman $c$ will never be part of the same blocking triple. Hence, in order for $M \circ N$ to be a 3D stable matching it suffices to ensure that for all $(b,c) \notin N$ we either have $b$ matched to someone better than $c$, meaning $N(b) \succ_b c$, or we have $c$ matched to some $N(c)$ such that $(N(c),b) \in \Rcal_c$. But this is exactly the definition of a stable matching for an instance of \SMG.  Hence we have the following theorem.

\begin{theorem}\label{thm0} \SE\ can be reduced in polynomial time to \SMG.
\end{theorem}
Note that Theorem~\ref{thm0}, together with Theorem~\ref{thm3} imply NP-hardness of SMG, and hence prove Theorem~\ref{thm3} (modulo containment in NP which is straightforward). Nevertheless, we choose to present first the proof of Theorem~\ref{thm1} as a warm-up, as it shares many similarity with that of Theorem~\ref{thm3}.

\section{Results for \SMG}

\subsection{\textsc{NP}-completeness of \SMG\ (Proof of Theorem \ref{thm1})}\label{sec3}

Containment in NP is straightforward, and is based on the observation that deciding whether an edge not in a perfect matching of an instance of \SMG is blocking or not can be done in polynomial time in $n$. The rest of our argument focuses on hardness.

Our proof uses a polynomial time reduction from \SMTI\ where ties occur only in the preference lists of the women. This problem is known to be \textsc{NP}-complete \cite{hard}. Let $\Ical$ be an instance of \textsc{SMTI} where ties 
occur only on the side of the women. We let $B= \lc b_1, \cdots, b_n \rc$ denote the set of men and $C = \lc c_1, \cdots, c_n \rc$ denote the set of women for the instance $\Ical$. For each person $q \in B \cup C$ we let $P(q)$ denote their preference list. 

We now describe how to construct an instance $\Jcal$ of \SMG. The set of men for our instance will be given by  $B' = B  \cup    \lc b_{n+1}  \rc$ and the set of women by $C' = C  \cup    \lc c_{n+1} \rc$. The preferences of the men are defined as 
follows: each original man $b \in B$ ranks the women in $P(b)$ first, in the same order as in $P(b)$, followed by the woman $c_{n+1}$, and the remaining women of $C'$ ranked arbitrarily; the man $b_{n+1}$ ranks the woman $c_{n+1}$ first, and the 
remaining women of $C'$ arbitrarily. Now, for each original woman $c \in C$ we define the binary relation $\Rcal_c \subseteq B \times B$ as follows $\Rcal_{c} := \lc (b,b') \, | \, b,b' \in P(c), b \succeq_{c} b' \rc$. That is, $\Rcal_c$ contains the ordered pair 
$(b,b')$ when both $b$ and $b'$ are acceptable to $c$ under the instance $\Ical$ and $c$ prefers $b$ at least as much as $b'$. 
Finally, for the woman $c_{n+1}$ we set $\Rcal_{c_{n+1}} := \emptyset$. This completes the definition of the instance $\Jcal$.

We use the introduction of the new agents to establish the following property
\begin{lemma}\label{lem2}  In any solution $N$ to $\Jcal$ every men $b \in B$ is matched to a woman from the set $P(b)$.
\end{lemma}

\begin{proof} Let $N$ be a solution to $\Jcal$. We first show that  $b_{n+1}$ is matched to $c_{n+1}$ in $N$. Suppose by contradiction that $N(b_{n+1}) \neq c_{n+1}$. Then from the way we defined the preferences of $b_{n+1}$ in $\Jcal$ we can conclude that $b_{n+1}$ prefers $c_{n+1}$ to its partner in $N$. Hence, in order for the pair $(b_{n+1},c_{n+1})$ to not be blocking with respect to the solution $N$ we must have $(N(c_{n+1}),b_{n+1} ) \in \Rcal_{c_{n+1}}$, which gives us a contradiction since we defined $\Rcal_{c_{n+1}} = \emptyset$. Hence $N(b_{n+1}) = c_{n+1}$.

We can now prove the lemma. Suppose by contradiction that there exists a man $b \in B$ such that $N(b) \notin P(b)$. Then $b$ cannot be matched to  $c_{n+1}$, since we showed that $c_{n+1}$ is always matched to $b_{n+1}$. From the way 
we defined the preferences of $b$ in $\Jcal$ we can conclude that $b$ prefers $c_{n+1}$ to its partner in $N$. Hence, in order for the pair $(b,c_{n+1})$ to not be blocking with respect to the solution $N$ we must have $(N(c_{n+1}),b ) \in \Rcal_{c_{n+1}}$, which gives us a contradiction since we defined $\Rcal_{c_{n+1}} = \emptyset$. \qed
\end{proof}

The following lemma completes the proof of Theorem \ref{thm1}.
\begin{lemma} $\Ical$ admits a perfect stable matching if and only if $\Jcal$ admits a stable matching.
\end{lemma}

\begin{proof} Suppose that $N$ is a perfect  stable matching for $\Ical$. Then complete $N$ to a perfect matching on $B' \cup C'$ by matching $b_{n+1}$ to $c_{n+1}$. To see that this is a stable matching for $\Jcal$, note that $b_{n+1}$ is matched to its most preferred woman  in $C'$, hence it cannot be part of any blocking pairs. It remains to show that no man in $B$ can be part of a blocking pair. Consider a man $b \in B$, and suppose that $c$ is a woman that $b$ strictly prefers to $N(b)$ according to the preferences in $\Jcal$. Then  it must be the case that $c \in P(b)$ and $b$ also strictly 
prefers $c$ to $N(b)$ in $\Ical$. Since $N$ is a solution to $\Ical$, woman $c$ prefers $N(c)$ at least as 
much as $b$ in $\Ical$. Hence from the way we defined the set $\Rcal_c$ we have $(N(c),b) \in \Rcal_c$, implying that $(b,c)$ is not blocking in $\Jcal$ either.

To see the other direction suppose that $\Jcal$ admits a stable matching, and let $N$ be the part of this stable matching obtained by restricting it to the sets $B \cup C$. It follows from Lemma \ref{lem2} that $N$ is a perfect matching and every man is matched to an acceptable woman. To see that there are no blocking pairs, consider any pair $(b,c) \notin N$ such that $(b,c)$ is an acceptable pair, that is $b \in P(c)$ and $c \in P(b)$. Assume now that in $\Ical$ man $b$ strictly prefers $c$ to $N(b)$. Since $c \in P(b)$ it follows that $b$ also strictly prefers $c$ to $N(b)$ in $\Jcal$. Hence, since $N$ is a stable matching for $\Jcal$, we must have $(N(c),b) \in \Rcal_c$. From the way we defined $\Rcal_c$ this implies that $N(c)$ is acceptable to $c$, and $c$ prefers $N(c)$ at least as much as $b$ in $\Ical$. Therefore $N$ does not have any blocking 
pairs, and is a stable matching in $\Ical$. \qed
\end{proof}

\subsection{Algorithmic results}\label{sec3.2}

In this section we introduce a  variant of the Gale-Shapley man-proposing algorithm for instances of \SMG that have asymmetric preferences. Let $\Ical$ be an \SMG instance as defined in Section \ref{sec2.2}. Like in the classical algorithm, each man in $B$ is originally declared single and is given a list containing all the women of $C$ in order of preference. In each round, every man $b$ that is still single proposes to its most preferred woman in $C$ that is still in his list. If a woman $c$ accepts a proposal from a man $b$ then they 
become engaged, and $b$'s  status changes from single to engaged. On the other hand, if $c$ rejects $b$'s proposal then $b$ removes woman $c$ from his list and remains single. The difference from the original Gale-Shapley algorithm is in the way that the women decide to accept or reject incoming proposals. A woman $c$ accepts a proposal from a man $b$ if and only if $(b,b') \in \Rcal_c$ for all other men $b'$ that have proposed to $c$ up to that point in the algorithm. This will ensure that whenever a woman $c$ rejects a proposal from a man $b$, $c$ is guaranteed to be matched at the end of the algorithm to some $b'$ such that $(b',b) \in \Rcal_c$ therefore ensuring that $(b,c)$ will not be a blocking pair. The description of the algorithm is given below.

\begin{algorithm}\caption{A deferred acceptance algorithm for 
\SMG}\label{alg}
\begin{algorithmic}[1]
\WHILE{there a single man in $B$ with a non-empty list}
\FORALL{$b$  single with a non-empty list} 
\STATE $b$ proposes to the top $c$ in its list
\ENDFOR
\FORALL{$b,c$ such that $b$ proposed to $c$}
\IF{ $(b,b') \in \Rcal_c$ for all $b' \neq b$ that proposed to $c$}
\STATE  $c$ accepts $b$ ( $(b,c)$ become an engaged pair)
\ELSE 
\STATE $c$ rejects $b$
\ENDIF
\ENDFOR
\ENDWHILE
\STATE If the set of engaged pairs forms a perfect matching return this solution, else conclude that $\Ical$ does not admit a stable matching. 
\end{algorithmic}
\end{algorithm}

It is easy to see that Algorithm 1 terminates and runs in polynomial time since each man in $B$ proposes to every woman in $C$ at most once. The following lemma is also easy to establish.

\begin{lemma}\label{lem3} Any solution returned by Algorithm 1 is a  stable matching.
\end{lemma}

\begin{proof} Let $N$ be a solution returned by Algorithm 1. It then follows that $N$ is a perfect matching. Suppose there is a blocking pair $(b,c)$. Then $b$ must have proposed to $c$ in Algorithm 1, and $c$ rejected $b$'s proposal. But, every time a woman $c$ rejects a proposal from $b$, Algorithm 1 ensures that $c$ will only become engaged to some man $b'$ such that $(b',b) \in \Rcal_c$. Therefore we cannot have any blocking pairs. \qed
\end{proof}

\noindent \textbf{Proof of Theorem \ref{thm2}.}  Using Lemma \ref{lem3} it suffices to show that if Algorithm 1 does not find a solution then the given instance does not admit a stable matching. Suppose by contradiction that there exists a stable matching $N$ but Algorithm 1 does not find a solution. First note that since the preferences are asymmetric, no woman $c \in C$ accepts more then one proposal at any point in the algorithm. Hence every woman is engaged to at most one man. Now since Algorithm 1 does not find a solution there is a man $b$ that is rejected by every woman in $C$. In particular, $b$ is rejected by $N(b)$. Among all pairs $(b,c) \in N$ such that $b$ proposed to $c$  and $c$ rejected $b$, let $(b_0,c_0)$ 
be the one that corresponds to the earliest rejection. As observed earlier, no man is rejected because of an arbitrary choice. Therefore, if $b_0$ was rejected at $c_0$ then there 
must be some man $b_1$ that also proposed to $c_0$ and $(b_0,b_1) \notin \Rcal_{c_0}$. Let $c_1$ be the partner of $b_1$ in $N$. We must have $c_1 \succ_{b_1} c_0$, since otherwise $(b_1,c_0)$ would be a blocking pair for the stable matching $N$. Thus $b_1$ proposed to $c_1$  before proposing to $c_0$ and $c_1$ rejected $b_1$ before $c_0$ rejected $b_0$. But this contradicts our choice of $(b_0,c_0)$. This concludes the proof of the theorem. \qed

\subsection{A polyhedral characterization}\label{sec3.3}

In this section we provide a polyhedral description for \SMG, that is an analogue of the well studied stable marriage polytope, first introduced in \cite{vande1989linear}. It is well known that the latter polytope is integral, meaning that the optimization version of \SM\ can be solved in polynomial time. For our setting, we show that our polytope can be used to efficiently decide the feasibility of an \SMG\ instance with asymmetric preferences, thus giving an  alternative proof of Theorem \ref{thm2}. We remark however that our polytope is not integral for this class of instances. Indeed, one can easily find instances with asymmetric preferences for which our polytope has fractional extreme points.

Given an instance $\Ical$ of \SMG\ we associated with each pair $(b,c) \in B \times C$ a variable $x_{bc}$, with the intended meaning that $x_{bc}=1$ if $b$ and $c$ are matched to each other and $x_{bc} =0$ otherwise. We then consider the following relaxation of the problem and let $P(\Ical)$ denote the set of all vectors  satisfying the constraints below
\begin{align}\label{P}
\sum_{c } x_{bc} &= 1 \quad \forall b  \in B\\
\sum_{b } x_{bc} &= 1 \quad \forall c  \in C\\
x_{bc} + \sum_{c' \succ_b c  } x_{bc'} + \sum_{(b',b) \in \Rcal_c} 
x_{b'c} &\geq 1 \quad \forall b \in B,  c \in C \\
x_{bc}&\geq 0 \quad \forall b \in B,c \in C 
\end{align}
It is easy to check that $x$ is the incidence vector of a stable matching for $\Ical$ if and only if $x$ is an integer vector in $P(\Ical)$. Our main result is the following.

\begin{theorem}\label{thm4} Let $\Ical$ be an instance of \SMG\ with asymmetric preferences. Then $P(\Ical) \neq \emptyset$ if and only if $\Ical$ admits a stable matching. Furthermore any fractional point $x \in P(\Ical)$ can be efficiently rounded to a stable matching solution for $\Ical$.
\end{theorem}

\begin{proof}  The first direction is trivial since if $\Ical$ admits a stable matching then the incidence vector corresponding to this stable matching is clearly in $P$. Now assume $P \neq \emptyset$ and let $x$ be any point in $P$.  We will show how to efficiently round $x$ to a stable matching, thus completing the proof of the theorem. For each $b \in B$ let $f(b)$ be $b$'s most preferred woman in the set $\lc c \in C : x_{bc} > 0 \rc$. Define $N= \lc (b,f(b)): b \in B \rc$. We first show that $N$ is a perfect matching. Since each man selects exactly one woman it suffices to show that no two men select the same woman. Suppose by contradiction that $f(b_1) = f(b_2) = c$ for some $b_1 \neq b_2$. Note that 
\ba
f(b) = c \quad &\Rightarrow \sum_{c' \succ_b c} x_{bc'} = 0 \quad 
\text{ from the definition of $f(b)$}\\
&\Rightarrow x_{bc} + \sum_{(b',b) \in \Rcal_c} x_{b'c} \geq 1 \quad 
\text{ from the stability constraint for $(b,c)$}\\
&\Rightarrow x_{bc} + \sum_{(b',b) \in \Rcal_c} x_{b'c} = 1 \quad 
\text{ from the matching constraint for $c$} \\
&\Rightarrow \sum_{(b',b) \notin \Rcal_c} x_{b'c} = 0.
\ea
Therefore $f(b_1) = f(b_2) = c$ implies that $(b_1,b_2) \in \Rcal_c$ and $(b_2,b_1) \in \Rcal_c$. But this contradicts the assumption that the preferences are asymmetric. Hence $N$ must be a perfect matching. To see that $N$ satisfies the stability constraints consider any pair $(b,c)$. If $f(b)=c$ or $f(b) \succ_b c$ then $(b,c)$ cannot be blocking, since $b$ will be matched in $N$ to someone he prefers at least as much as $c$. Hence it suffices to consider the case where $c \succ_b f(b)$. But then we must have $x_{bc} + \sum_{c' \succ_b c} 
x_{bc'} = 0$ and since $x \in P(\Ical)$ this implies that $\sum_{(b',b) \in \Rcal_c } x_{b'c}=1$. Now, since $N$ uses only edges in the support of $x$, it follows that $(N(c),b) \in \Rcal_c$  and hence $(b,c)$ is not blocking. Therefore $N$ is a stable matching for $\Ical$. \qed
\end{proof}

\section{\textsc{NP}-completeness of \SE\ (Proof of Theorem \ref{thm3}) }\label{sec4}

In this section we provide the proof of Theorem \ref{thm3}.  At a high level, our proof adopts a similar strategy as that of Theorem~\ref{thm1}. In particular, we will provide a polynomial-time reduction for \SMTI\ to \SE. However, due to the additional structural properties of \SE\ instances  this reduction becomes significantly more intricate. The full proof can be found below.

\begin{proof} 
Containment of \SE\ in NP follows by observing that deciding whether any triple not in a 3D matching of any instance (extending the given 2D matching) is blocking or not, can be done in polynomial time in $n$. The rest of the argument focuses on hardness. Let $\Ical$ be an instance of \SMTI\ where ties occur only on the side of the women. We let $B= \lc b_1, \cdots, b_n \rc$ denote the set of men and $C = \lc c_1, \cdots, c_n \rc$ denote the set of women for the instance $\Ical$. For each person $q \in B \cup C$ let $P(q)$ denote their preference list. For a given woman $c_i \in C$ the preference list $P(c_i)$ might have several men tied at a certain position. We let $t_i$ be the number of positions in $c_i$'s preference list, so that $1 \leq t_i \leq n$. Now for each $j \in \lc 1, \cdots, t_i \rc$ we let $P(c_i)_j$ denote the men that are tied in position $j$ in $c_i$'s preference 
list. Furthermore, for each $b \in P(c_i)$ we let $r_{c_i}(b)$ denote the position that $b$ occupies in $P(c_i)$.

We now describe how to construct an instance $\Jcal$ of \SE. Our instance will consist of three sets $A'$, $B'$ and $C'$. We will include all the elements of $B$ in  $B'$, and all the elements of $C$ in $C'$. We will also define some extra  elements as follows
\bi
\item[--] For all $i \in \lc 1,2, 3 \rc$ we create the new elements $a_{n+i}$, $b_{n+i}$ and $c_{n+i}$.
\item[--] For all $i \in [n]$ we create $t_i$ new elements $B_i = \lc b_{i,1}, \cdots, b_{i,t_i} \rc$, together with the corresponding sets $A_i = \lc a_{i,1}, \cdots, a_{i,t_i} \rc$, and $C_i = \lc c_{i,1}, \cdots, c_{i,t_i} \rc $.
\item[--] We create a set of $n$ new elements  $A = \lc a_1, \cdots, a_n \rc$, corresponding to the original sets $B$ and $C$.  
\ei

We then set
\ba
A' &= A \cup  \bigcup_{i=1}^{n} A_i  \cup  \lc a_{n+1}, a_{n+2}, 
a_{n+3}  \rc\\
B' &= B  \cup  \bigcup_{i=1}^{n} B_i  \cup  \lc b_{n+1}, b_{n+2}, 
b_{n+3}  \rc\\
C' &= C  \cup  \bigcup_{i=1}^{n} C_i   \cup  \lc c_{n+1}, c_{n+2}, 
c_{n+3}  \rc
\ea
We now fix the perfect matching $M$ on $A' \times B'$ by matching each dog of $A'$ to its corresponding man in $B'$. That is $a_i$ is matched to $b_i$ for all $i \in [n]$, $a_{i,j}$ is matched to $b_{i,j}$ for all $i \in [n]$ and $j \in [t_i]$, and $a_{n+i}$ is matched to $b_{n+i}$ for all $i \in \lc 1,2,3 \rc$.

Next, we create the preferences lists of each agent. For each $q \in A' \cup B' \cup C'$ we let $P'(q)$ denote the preference list of agent $q$ for the instance $\Jcal$. Recall that for each $a \in A'$ the preference list $P'(a)$ must be a complete and strict ordering of the set $B'$. Similarly for each $b \in B'$ the list $P'(b)$ will be a complete and strict ordering of the set $C'$, and for each $c \in C'$ the list $P'(c)$ will be a complete and strict ordering of the set $A'$. Table \ref{table} summarizes the preferences of each agent. 
When listing the preference list $P'(q)$ of an agent $q$, we use the notation $[Q]$ to denote that the agents of the set $Q$ appear in consecutive positions in $P'(q)$ in any arbitrary order among them. The notation $[\cdots]$ is used to denote that the remaining agents that have not been listed in $P'(q)$ appear in consecutive positions in $P'(q)$ in any arbitrary order among them.

The preferences of the dogs in $A'$ over the set of men $B'$ are defined as follows  
\bi
\item[--] For all $i \in [n]$: $a_i$ ranks $b_i$ first, and the remaining men of $B'$ arbitrarily.
\item[--] For all $i \in [n]$ and for all $j \in [t_i]$: $a_{i,j}$ ranks the men  that are tied in position $j$ of women $c_i$'s list (in the instance $\Ical$) at the top of its list in any arbitrary order among them, followed by the man $b_{i,j}$. The remaining men of $B'$ are ranked arbitrarily.
\item[--] $a_{n+1}$ ranks $b_{n+1}$ first, $a_{n+2}$ ranks $b_{n+2}$ last, and $a_{n+3}$ ranks $b_{n+2}$ first followed by $b_{n+3}$ second. The rest of their lists are arbitrary.
\ei

\begin{table}
\begin{center}
    \begin{tabular}{ | c | c |}
    \hline
    &  $P'(a_i) = b_i, [\cdots]$ \\  
    $ \forall i \in [n]$ & $P(b_i), c_{n+1}, [\cdots]$ \\  
    & $P'(c_i) = [M\lb P(c_i)_1 \rb], a_{i,1},  \cdots, [M\lb 
P(c_i)_{t_i} \rb], a_{i,t_i}, [\cdots]$\\ \hline
    &  $P'(a_{i,j}) = [P(c_i)_j], b_{i,j}, [\cdots]$ \\  
    $ \forall i \in [n], $ & $P'(b_{i,j}) = c_{i,j}, [\cdots]$ \\  
    $\forall j \in [t_i]$ & $P'(c_{i,j}) =  a_{n+2},[\cdots]$ \\  
\hline 
    & $P'(a_{n+1})= b_{n+1},  [\cdots]$ \\
    & $P'(a_{n+2})= [\cdots], b_{n+2} $ \\
    & $P'(a_{n+3})= b_{n+2},b_{n+3},  [\cdots]$ \\ \hline 
    $\forall i \in [3]$ & $P'(b_{n+i})= c_{n+i},[\cdots]$  \\ \hline 
    & $P'(c_{n+1})= a_{n+2},[\cdots]$ \\
    & $P'(c_{n+2})= a_{n+3},[\cdots]$ \\
    & $P'(c_{n+1})= a_{n+2},[\cdots]$ \\ \hline
    \end{tabular}
\caption{The preferences for the instance $\Jcal$.} \label{table}
\end{center}
\end{table}

The preferences of the men in $B'$ over the set of women $C'$ are defined as follows
\bi
\item[--] For all $i \in [n]$: $b_i$ ranks the women in $P(b_i)$ first, in the same order as in $P(b_i)$, followed by the woman $c_{n+1}$, and the remaining women of $C'$ ranked arbitrarily.
\item[--] For all $i \in [n]$ and for all $j \in [t_i]$: $b_{i,j}$ ranks $c_{i,j}$ first, and the remaining women of $C'$ arbitrarily.
\item[--] For all $i \in \lc 1,2, 3 \rc$: $b_{n+i}$ ranks $c_{n+i}$ first and the remaining women of $C'$ arbitrarily.
\ei

The preferences of the women in $C'$ over the set of dogs $A'$ are defined as follows
\bi
\item[--] For all  $i \in [n]$:  $c_i$ ranks the  dogs in the set $M\lb P(c_i)_1 \rb$ at the top of its list in an arbitrary order, followed by the dog $a_{i,1}$, then the dogs in the set $M\lb P(c_i)_2 \rb$ again in an arbitrary order, followed by the dog $a_{i,2}$, and so on until the dogs in  the set $M\lb P(c_i)_{t_i} \rb$ followed by the  $a_{i,t_i}$. The rest of $c_i$'s preference list is arbitrary.
\item[--] For all $i \in [n]$ and for all $j \in [t_i]$: $c_{i,j}$ ranks $a_{n+2}$ first, and the remaining dogs in $A'$ arbitrarily. 
\item[--] $c_{n+1}$ and $c_{n+3}$ rank $a_{n+2}$ first, and the remaining dogs in $A'$ arbitrarily, while $c_{n+2}$ ranks $b_{n+3}$ first and the remaining dogs in $A'$ arbitrarily.
\ei

This concludes the description of the instance $\Jcal$. Note that since we already fixed the perfect matching $M$, the sets $\lc A_b: b \in B' \rc$ are now fully determined as follows
\begin{align}\label{setAb}
A_b = \begin{cases} 
\lc a_{i,r_{c_i}(b)} : i \in [n], b \in P(c_i) \rc \cup \lc a_{n+2} \rc &\text{ if $b \in B$,} \\ 
\lc a_{n+2} \rc &\text{ if $b \in B' \bb \lb B \cup \lc b_{n+2} \rc \rb, $} \\
\lc a_{n+3} \rc &\text{ if $b= b_{n+2}$.} \\
\end{cases}
\end{align}

This also fixes the sets $\Rcal_c$ for all $c \in C$ using the definition given in equation \eqref{set} of Section \ref{se}. In particular, the sets $\Rcal_c$ satisfy the following property.

\begin{lemma}\label{lem4} For every  $c \in C$ and $b,b' \in B$  where $b$ is an acceptable partner of $c$ in $\Ical$, we have $(b',b) \in \Rcal_c$ if and only if $b'$ is also an acceptable partner of $c$ in $\Ical$ and $c$ prefers $b'$ at least as much as $b$.
\end{lemma}

\begin{proof} Let $(b,c) \in B \times C$ such that $(b,c)$ is an acceptable pair for the instance $\Ical$. Then $b$ must appear in the preference list $P(c)$ of woman $c$. Recall that $r_c(b)$ denotes the position that $b$ occupies in $P(c)$. Then it follows from the description of $A_b$ in \eqref{setAb} that the only dog from the set $\lc a_{i,1}, \cdots, a_{i,t_i} \rc$ that is contained in $A_b$ is  $a_{i,r_c(b)}$. It then follows from the way we defined the preference list $P'(c_i)$ that $\alpha(b,c_i) = a_{i,r_c(b)}$. The 
lemma then follows form the definition of the sets $\Rcal_c$. \qed
\end{proof}

The proof of the Theorem is completed via the following Lemmas.

\begin{lemma}\label{lem5} In any solution to $\Jcal$ $b_{n+2}$ and $b_{n+3}$  are matched to  $c_{n+2}$ and $c_{n+3}$. 
\end{lemma}

\begin{proof} Let $N$ be a solution to $\Jcal$. Suppose by contradiction that this is not true. Then at least one of the men $b_{n+2}$ and $b_{n+3}$ must be matched to a woman that is neither  $c_{n+2}$ nor $c_{n+3}$. Assume that this is $b_{n+2}$. The case with $b_{n+3}$ follows from a symmetric argument. Then since $c_{n+2}$ was the most preferred woman of $b_{n+2}$ it follows that $c_{n+2}$ must be matched to some $b'$ such that $(b',b_{n+2}) \in \Rcal_{c_{n+2}}$, since otherwise the pair $(b_{n+2},c_{n+2})$ would be blocking. Now from \eqref{setAb} we have that $A_{b_{n+2}} = \lc a_{n+3} \rc$ hence $\alpha(b_{n+2}, c_{n+2}) = a_{n+3}$ and therefore it follows from the preference list $P'(c_{n+2})$ that $(b',b_{n+2}) 
\in \Rcal_{c_{n+2}}$ if and only if $b=b_{n+3}$. Hence $c_{n+2}$ must be matched to $b_{n+3}$. But now we can argue that since $c_{n+3}$ was the most preferred woman of $b_{n+3}$,  $c_{n+3}$ must be matched to some $b'$ such that $(b',b_{n+3}) \in \Rcal_{c_{n+3}}$. Again from \eqref{setAb} we have that $A_{b_{n+3}} = \lc a_{n+3} \rc$ hence $\alpha(b_{n+3}, c_{n+3}) = a_{n+2}$ and therefore it follows from the preference list $P'(c_{n+3})$ that  $(b',b_{n+3}) \in \Rcal_{c_{n+3}}$ if and only if $b' = b_{n+2}$. But this means that 
$c_{n+3}$ must be matched to $b_{n+2}$, which contradicts our assumption that $b_{n+2}$ is not matched to neither $c_{n+2}$ nor $c_{n+3}$. \qed
\end{proof}

\begin{lemma}\label{lem6} In any solution to $\Jcal$ $b_{n+1}$ is matched to $c_{n+1}$.
\end{lemma}

\begin{proof} Let $N$ be a solution to $\Jcal$. Suppose by contradiction that $b_{n+1}$ is not matched to $c_{n+1}$. Then since $c_{n+1}$ was the most preferred woman of $b_{n+1}$ it follows that $c_{n+1}$ must be matched to some $b'$ such that $(b',b_{n+1}) \in \Rcal_{c_{n+1}}$. Now from \eqref{setAb} we have that $A_{b_{n+1}} = \lc a_{n+2} \rc$ which implies that $\alpha(b_{n+1}, c_{n+1}) = a_{n+2}$. Therefore it follows from the preference list $P'(c_{n+1})$ that  $(b',b_{n+1}) \in \Rcal_{c_{n+1}}$ if and only if $b' = b_{n+2}$. But from Lemma \ref{lem5} we know that  $b_{n+2}$ is always matched to either $c_{n+2}$ or $c_{n+3}$. Therefore by contradiction, $b_{n+1}$ must be matched to $c_{n+1}$. \qed
\end{proof}

\begin{lemma}\label{lem7} In any solution to $\Jcal$ every man $b \in B$ is matched to a woman from the set $P(b)$.
\end{lemma}

\begin{proof} Consider a man $b \in B$ and suppose by contradiction that $b$ is not matched to someone in $P(b)$. It follows from Lemma \ref{lem6} that $b$ cannot be matched to $c_{n+1}$, since $c_{n+1}$ is matched to $b_{n+1}$. From the way we defined the preferences of $b$ in $\Jcal$ we can conclude that $b$ prefers $c_{n+1}$ to its match in $N$. Furthermore we have  $a_{n+2} \in A_b$ from \eqref{setAb}. Now since $c_{n+1}$ ranks $a_{n+2}$ at the top of its list it follows that $\alpha(b, c_{n+1}) = a_{n+2}$ and thus $(b',b) \in \Rcal_{c_{n+1}}$ if and only if $b=b_{n+2}$. But then the pair $(b,c_{n+1})$ is blocking, thus contradicting the fact that  $N$ is  a solution to $\Jcal$. \qed
\end{proof}

\begin{lemma} $\Ical$ admits a complete weakly stable matching if and only if $\Jcal$ admits a stable extension.
\end{lemma}

\begin{proof} Suppose that $N$ is a complete  stable matching for $\Ical$. Then complete $N$ to a perfect matching on $B' \cup C'$ by matching $b_{i,j}$ to $c_{i,j}$ for every $i \in [n]$ and $j \in [t_i]$ and matching $b_{n+i}$ to $c_{n+i}$ for every $i \in \lc 1,2,3 \rc$. Note that every man in $B' \bb B$ is matched to the woman that it prefers the most 
hence no man in $B'$ can be part of a block. From Lemma \ref{lem7} we have that every man in $B$ is matched to a partner that is acceptable in $\Ical$. Hence for any $b \in B$, if $c$ is a woman that $b$ strictly prefers to $N(b)$ according to the preferences in $\Jcal$, then it must be the case that $c \in P(b)$ and $b$ also strictly prefers $c$ to $N(b)$ in $\Ical$. Since $N$ is a solution to $\Ical$ it follows that $c$ prefers $N(c)$ at least as much as $b$. We can now use Lemma \ref{lem4} to conclude that $(N(c),b) \in \Rcal_c$. 
Therefore $(b,c)$ cannot be a blocking pair in $\Jcal$. Hence the perfect matching that we defined from $N$ is solution to the instance $\Jcal$.

Conversely, suppose that $\Jcal$ admits a stable extension, and let $N$ be the part of this stable extension obtained by restricting it to the sets $B \cup C$. It follows from Lemma 
\ref{lem7} that $N$ is a perfect matching and every agent is matched to an acceptable partner. To see that there are no blocking pairs consider any pair $(b,c) \notin N$ such that $(b,c)$ is an acceptable pair, that is $b \in P(c)$ and $c \in P(b)$. Assume now that in $\Ical$ man $b$ strictly prefers $c$ to $N(b)$. Since $c \in P(b)$ it follows that $b$ also strictly prefers $c$ to $N(b)$ in $\Jcal$. Since $N$ is a solution to $\Jcal$ we must have $(N(c),b) \in \Rcal_c$. Again using Lemma \ref{lem4} we can conclude that $N(c)$ is a man in $P(c)$ that $c$ prefers at least as much as $b$ in $\Ical$. Therefore $N$ does not have any blocking pairs, and isa weakly stable matching in $\Ical$. \qed
\end{proof}

\end{proof}

\medskip
\noindent \textbf{Acknowledgement.} We thank the reviewers for their thorough and careful review. We  highly appreciate their insightful suggestions that led to a substantial improvement of the paper.

\end{document}